\documentclass[conference]{IEEEtran}
\IEEEoverridecommandlockouts
\usepackage{amssymb}
\usepackage{amsfonts}
\usepackage{amsmath}
\usepackage{amsthm}
\usepackage{blkarray}
\usepackage{blindtext}
\usepackage{graphicx}
\usepackage{booktabs}
\usepackage{dsfont}
\usepackage{balance}
\usepackage{color}
\usepackage{cite}
\usepackage{enumerate}
\usepackage{footnote}
\usepackage{verbatim}
\usepackage{stfloats}
\usepackage{mdframed}
\usepackage{multirow}
\usepackage{subfigure}
\usepackage{tablefootnote}
\usepackage{tabularx} 
\usepackage{textcomp}
\usepackage{float}
\usepackage{xr}
\usepackage{xcolor}
\usepackage[final]{hyperref} 
\hypersetup{hidelinks}
\usepackage[lined, boxed, linesnumbered, ruled]{algorithm2e}
\newtheorem{theorem}{Theorem}[section]
\newtheorem{corollary}[theorem]{Corollary}
\newtheorem{definition}[theorem]{Definition}

\newtheorem{proposition}[theorem]{Proposition}

\usepackage[utf8]{inputenc}

\title{Message-Level Scheduling for RLNC-Coded Multi-Source Traffic}

\author{
  \IEEEauthorblockN{Zhaohong Lu\IEEEauthorrefmark{1}, Qingyu Liu\IEEEauthorrefmark{2},Haibo Zeng\IEEEauthorrefmark{1}}
  \IEEEauthorblockA{\IEEEauthorrefmark{1} Dept. of Electrical and Computer Engineering, Virginia Tech, Blacksburg, VA 24060, USA\\
                    \IEEEauthorrefmark{2}Shenzhen Graduate School, School of Electric and Computer Engineering, Peking University, Shenzhen 518055, China\\
                    Emails: \{zhaohonglu,hbzeng\}@vt.edu, qy.liu@pku.edu.cn
                    }
}

\begin{document}

\maketitle

\begin{abstract}
This paper studies weighted decoding-delay minimization for multiple
RLNC-coded message streams that compete for finite processing capacity at
a destination. Packet arrivals are exogenous, while the scheduler only
determines the processing order of packets already available at the
destination. A trace-conditioned offline scheduling formulation shows
that a batch-release subclass is strongly NP-hard even with a single
processing unit. Message-Aware Innovation-Deficit Scheduling (MAIDS) is
then developed to prioritize each serviceable message according to its
weight and remaining decoding deficit. For a single processing unit,
MAIDS is shown to be exactly optimal under nonblocking progressive arrivals
with equal weights and under common activation with arbitrary positive
weights, while the unrestricted weighted online problem admits no universal
deterministic $O(1)$ competitive ratio. Simulation results on streaming and batch benchmarks show that MAIDS consistently reduces weighted decoding delay relative to the tested baselines,
remains close to the offline optimum on average, and recovers the predicted
exact performance boundaries.
\end{abstract}

\section{Introduction}
\label{sec:int3}

Vehicular ad hoc networks (VANETs) support direct vehicle-to-vehicle
and vehicle-to-infrastructure communication for safety, cooperative
driving, traffic efficiency, and infotainment
\cite{4539481,5948952}. Their high mobility, limited communication
range, and rapidly varying wireless links lead to frequent topology
changes and make stable end-to-end routes difficult to maintain
\cite{1364012}. Broadcast and opportunistic forwarding are therefore
natural dissemination mechanisms: information can propagate through
whichever relays are available, rather than relying on a persistent
source--destination path. These characteristics motivate multi-source
packet arrivals at a destination, because packets belonging to the same
message may be forwarded through different relays and reach the receiver
at different times.

Random linear network coding (RLNC) is well suited to such dissemination.
Instead of forwarding only native packets, a node transmits coded
combinations of packets from the same message generation, and a
destination can recover the message after accumulating the required
number of innovative degrees of freedom (DoFs)
\cite{ho2006rlnc,chachulski2007more}. Intermediate relays may also
recode received packets. Consequently, from the destination's
perspective, an upstream packet source may be either the original
message generator or a relay, and multiple upstream nodes may
simultaneously contribute coded packets for the same message. Different
sources may also serve different messages, producing overlapping
message streams at the receiver.

Our previous work studied the \emph{network-side} delay of
network-coding-enabled multi-source dissemination, including
propagation, relay queueing, path diversity, and congestion effects
\cite{lu2025delayanalysisrandomnetwork}. That analysis asks how coded
packets travel through the network and how network structure determines
delivery delay. The present work addresses the complementary problem
that begins \emph{after reception}: once coded packets from several
messages have reached the same destination, limited receiver processing
capacity can itself become a bottleneck. A destination may therefore
hold several unfinished messages simultaneously, even when the upstream
network has already delivered useful coded information for all of them.
The receiver must decide which available packets to process first in
order to reduce message-level decoding delay.

Accordingly, the problem is formulated at the message level, where
multiple RLNC-coded messages arrive over time at a destination with
finite processing capacity. In each slot, the scheduler determines how
the available processing capacity is allocated among coded packets that
have already arrived, with the objective of minimizing weighted decoding
delay from message activation to successful decoding. The upstream
transmission process is outside the scheduler's control: each source
serves one current message at a time, repeatedly transmitting coded
packets of that message and switching only after completing its assigned
transmissions, while different sources may contribute to the same or to
different messages. The scheduler cannot select transmitters, alter
future packet arrivals, or feed scheduling decisions back upstream.
Under the rank-unit abstraction used in this paper, each successfully
received coded packet contributes one innovative DoF until the message
reaches rank $K_m$. Thus, the scheduling state is naturally described at
the message level by its weight, remaining decoding deficit, and
currently available innovative work.

Although VANETs provide the primary motivating example, the theoretical
model is intentionally more general. Mobility, forwarding, link losses,
and source activity are summarized by an exogenous coded-packet arrival
trace at one receiver. The resulting scheduling problem therefore
applies to general RLNC-coded multi-source traffic satisfying the same
control boundary, rather than to VANETs alone.

This problem differs from both network-side coded-packet
scheduling and decoder-internal computation scheduling. Prior work has
optimized coding and transmission decisions before reception
\cite{eryilmaz2008delay,huang2015distributed,garcia2017lowdelay,
skevakis2019scheduling,wang2019prioritized,tasdemir2020platooning},
whereas our arrival trace is exogenous. Progressive multicore RLNC
decoding has also been studied by scheduling Gauss--Jordan operations
within a coded generation \cite{wunderlich2019progressive}; our
scheduler instead allocates processing service across multiple incomplete
messages. The closest classical connection is preemptive weighted
scheduling with release times. Ratio-based online policies and WSRPT
have been studied for weighted completion objectives
\cite{schulz2002alphapoints,batsyna2014online,xiong2012wsrpt}, while our
decoding-delay objective is a weighted-flow-time objective. This
distinction matters because unrestricted online weighted flow time
admits no deterministic $O(1)$ competitive ratio
\cite{bansal2009weightedflow}.

The main contributions of this paper are:
\begin{itemize}
    \item A general scheduling model is formulated for continuous
    multi-source RLNC traffic with exogenous packet arrivals, finite
    processing capacity, and weighted message decoding delay.

    \item A finite trace-conditioned offline benchmark is constructed,
    and a deterministic batch-release subclass is shown to be strongly
    NP-hard even with one processing unit.

    \item Message-Aware Innovation-Deficit Scheduling (MAIDS) is developed
    to prioritize serviceable messages by weight divided by remaining
    decoding deficit. For a single processing unit and nonblocking
    progressive arrivals, MAIDS is exactly optimal for equal weights via
    SRPT equivalence and for common activation with arbitrary positive
    weights via Smith's ratio rule. In contrast, the unrestricted
    weighted online problem has no deterministic $O(1)$ competitive
    guarantee.

\end{itemize}

The remainder of the paper is organized as follows.
Section~\ref{sec:system_model} introduces the receiver model and formulates the
optimization objective, after which Section~\ref{sec:offline} establishes the
trace-conditioned offline benchmark and its complexity. MAIDS and its
performance boundaries are developed in Section~\ref{sec:general_scheduler},
with numerical evaluation presented in Section~\ref{sec:simulation}.

\section{System Model and Problem Formulation}
\label{sec:system_model}

\subsection{Streaming System}

Consider a slotted-time system with one focal destination $v$
receiving RLNC-coded traffic from a fixed set of upstream sources
$\mathcal N=\{1,\ldots,N\}$. Multiple messages are delivered continuously,
and a source may be either an original generator or a relay. The same
message may therefore be supplied by more than one source, while different
sources may simultaneously carry different messages. Packets received at
the beginning of a slot are immediately available for processing in that
slot.

For each message, one or more upstream sources may be assigned to generate
or forward its coded packets. Once source $n$ is assigned to a message, it
repeatedly transmits coded packets for that message until its assigned
number of transmissions is completed. The source then becomes available
for reassignment to a subsequent message. The packet stream observed at
the destination is thus the superposition of transmissions from multiple
sources, whose message assignments may evolve over time.

Several source-side constraints govern this process. Each source is
assigned to at most one current message at a time, so transmissions of
successive messages do not interleave at an individual source. Different
sources, however, may concurrently contribute to the same message or to
different messages. Source assignment, transmission timing, forwarding,
and message rollover are all upstream processes and remain outside the
scheduler's control. Once coded packets reach the destination, the
scheduler therefore treats the resulting packet-arrival sequence as
exogenous; the analytical model does not require a particular stochastic
transmission law, while a concrete transmission process is introduced
later for the streaming simulations.

This separation allows the traffic seen by the scheduler to be described
directly in terms of message activations. Let $\mathcal M^{\rm new}(t)$
denote the set of messages whose first coded packet is observed at the
destination in slot $t$. For each message $m$, its activation time is
\begin{equation}
    a_m
    =
    \min\{t:m\in\mathcal M^{\rm new}(t)\}.
    \label{eq:message_activation_time}
\end{equation}
Once activated, message $m$ enters $\mathcal M(t)$, the set of active
not-yet-decoded messages, with generation size $K_m$ and positive weight
$w_m$, and remains there until decoding. Because new messages may be
activated before earlier ones complete, multiple unresolved messages can
coexist at the destination, including messages served successively by the
same upstream source.

\subsection{RLNC Arrivals and Receiver State}

The resulting packet arrivals are represented by a message-level decoding
state at the destination.

Each message $m$ contains $K_m$ source DoFs and requires $K_m$
innovative DoFs for decoding. To focus on scheduling rather than
finite-field dependence, a rank-unit abstraction is adopted: every
successfully received coded packet contributes one innovative DoF until
the message reaches rank $K_m$, after which additional packets provide
no further decoding progress. Coding is confined within a single
message generation.

Let $\mathcal A_m(t)$ denote the set of coded packets of message $m$
that become available at the destination at the beginning of slot $t$,
possibly aggregated from multiple upstream sources. Because the
scheduler does not influence upstream transmission or forwarding, the
arrival trace $\{\mathcal A_m(t)\}$ is treated as exogenous. It may
therefore reflect source activity, forwarding, mobility, lower-layer
access, wireless loss, and other reception effects without requiring
those mechanisms to be modeled explicitly.

For an active message $m$, let $r_m(t)$ be the number of innovative
DoFs already incorporated before slot-$t$ processing. The corresponding
remaining decoding deficit is
\begin{equation}
    d_m(t)=[K_m-r_m(t)]^+.
    \label{eq:rank_deficit}
\end{equation}
Let $\mathcal Q_m(t)$ denote the unprocessed packet buffer carried into
slot $t$. After the new arrivals are included, the available buffer is
\begin{equation}
    \mathcal Q_m^+(t)
    =
    \mathcal Q_m(t)\cup\mathcal A_m(t).
    \label{eq:buffer_after_arrival}
\end{equation}
Under the rank-unit abstraction, at most the remaining deficit can be
usefully processed. Hence the serviceable work of message $m$ in slot
$t$ is
\begin{equation}
    b_m(t)
    =
    \min\!\left\{d_m(t),\,|\mathcal Q_m^+(t)|\right\},
    \label{eq:buffered_innovation}
\end{equation}
with $0\leq b_m(t)\leq d_m(t)$.

\subsection{Receiver Processing and Control Boundary}
\label{sec:information_structure}

In each slot, the scheduler allocates finite processing capacity among
active messages. Let $C$ denote the maximum number of coded packets that can be
processed per slot, and let $y_m(t)\in\mathbb Z_{\geq0}$ be the number
of innovative packets of message $m$ processed in slot $t$. A feasible
allocation must satisfy
\begin{equation}
    \sum_{m\in\mathcal M(t)}y_m(t)\leq C,
    \label{eq:rx_processing_capacity}
\end{equation}
and
\begin{equation}
    0\leq y_m(t)\leq b_m(t),
    \qquad m\in\mathcal M(t).
    \label{eq:buffer_service_constraint}
\end{equation}
Processing $y_m(t)$ packets advances the incorporated rank by the same
amount, so the message state evolves according to
\begin{equation}
    r_m(t+1)=r_m(t)+y_m(t),
    \qquad
    d_m(t+1)=[d_m(t)-y_m(t)]^+.
    \label{eq:deficit_evolution}
\end{equation}
Thus $y_m(t)$ is the scheduler's control variable, while message
generation, source activity, transmitter identity, packet arrival times,
wireless reception, and the parameters $K_m$, $w_m$, and $C$ remain
outside its control.

\subsection{Message Decoding Delay and Objective}

The state and service dynamics above determine when each active message
can be decoded. Message $m$ completes in the first slot after which its
remaining decoding deficit becomes zero:
\begin{equation}
    T_m^{\rm dec}
    =
    \min\{t\geq a_m:d_m(t+1)=0\}.
    \label{eq:message_completion_slot}
\end{equation}
The corresponding decoding delay, measured from activation through the
completion slot, is
\begin{equation}
    \Delta_m=T_m^{\rm dec}-a_m+1.
    \label{eq:message_delay}
\end{equation}
For any finite set $\mathcal S$ of evaluated messages, the scheduling
objective is the total weighted decoding delay
\begin{equation}
    J(\mathcal S)
    =
    \sum_{m\in\mathcal S}w_m\Delta_m.
    \label{eq:weighted_delay_objective}
\end{equation}
The scheduler operates continuously and requires neither a terminal
batch nor knowledge of future arrivals.

Throughout the analysis, coded packets are assumed to be of fixed size,
receiver buffering is sufficient, and packet processing is deterministic
and homogeneous.
For steady-state evaluation, $\rho$ denotes the long-run average
innovative-packet workload offered to the receiver, and only the stable
regime $\rho/C<1$ is considered. Inter-message coding and joint
transmitter--receiver scheduling are outside the scope of this work.

\section{Trace-Conditioned Offline Benchmark and Complexity}
\label{sec:offline}

The operational system contains a continuing message stream, whereas
exact offline optimization necessarily considers a finite
trace-conditioned instance. Accordingly, a finite set
$\mathcal M_H$ of messages is selected from the stream together with
their realized activation times and packet-arrival traces, and the
horizon $H$ is chosen large enough for all messages in $\mathcal M_H$
to be feasibly completed. The offline receiver is omniscient with respect to
this finite trace, but the control boundary is unchanged:
upstream message generation and packet arrivals remain exogenous.

For each $m\in\mathcal M_H$, define the cumulative innovative rank
made available from its activation through slot $t$ as
\begin{equation}
    R_m^{\rm arr}(t)
    =
    \min\!\left\{
        K_m,
        \sum_{\tau=a_m}^{t}|\mathcal A_m(\tau)|
    \right\},
    \qquad t\geq a_m.
    \label{eq:cumulative_arrival_rank}
\end{equation}
Under the rank-unit abstraction, $R_m^{\rm arr}(t)$ is the number of
innovative DoFs that have become available at the destination by slot
$t$, capped at the decoding requirement $K_m$, and is independent of
the receiver's later processing order. The selected horizon $H$ is
assumed sufficiently long for every evaluated message to be decodable,
i.e.,
\begin{equation}
    R_m^{\rm arr}(H)\geq K_m,
    \qquad m\in\mathcal M_H.
    \label{eq:offline_decodability}
\end{equation}

\subsection{Offline Optimization Formulation}

For a newly activated message, the required receiver work is $K_m$
innovative packets. Let $y_m(t)\in\mathbb Z_{\geq0}$ denote the number of
innovative packets of message $m$ processed in slot $t$, and let
$z_{m,t}\in\{0,1\}$ designate slot $t$ as its completion slot. Define
\begin{equation}
    Y_m(t)
    \triangleq
    \sum_{\tau=a_m}^{t}y_m(\tau),
    \qquad
    Z_m(t)
    \triangleq
    \sum_{\tau=a_m}^{t}z_{m,\tau}.
    \label{eq:cumulative_service_completion}
\end{equation}
All message-specific constraints below hold for
$m\in\mathcal M_H$ and $t=a_m,\ldots,H$.

The complete finite trace-conditioned problem is
\begin{subequations}
\label{eq:offline_problem}
\begin{align}
\min_{\mathbf y,\mathbf z}\quad
&
\sum_{m\in\mathcal M_H}
\sum_{t=a_m}^{H}
w_m (t-a_m+1) z_{m,t}
\label{eq:offline_obj}
\\
\mathrm{s.t.}\quad
&
Y_m(t)
\leq
R_m^{\rm arr}(t),
\label{eq:arrival_causality}
\\
&
\sum_{\substack{m\in\mathcal M_H:\\ a_m\leq t}}
y_m(t)
\leq
C,
\qquad t=1,\ldots,H,
\label{eq:offline_processing_capacity}
\\
&
Y_m(H)
=
K_m,
\label{eq:offline_total_service}
\\
&
Z_m(H)
=
1,
\label{eq:completion_once}
\\
&
K_m Z_m(t)
\leq
Y_m(t),
\label{eq:completion_consistency}
\\
&
y_m(t)\in\mathbb Z_{\geq0},
\qquad
z_{m,t}\in\{0,1\}.
\label{eq:offline_domains}
\end{align}
\end{subequations}

Constraint~\eqref{eq:arrival_causality} enforces destination-side
information availability. Constraint~\eqref{eq:offline_processing_capacity}
limits aggregate receiver processing to $C$ coded packets per slot.
Constraint~\eqref{eq:offline_total_service} requires every considered
message to receive all $K_m$ decoding DoFs. Constraint
\eqref{eq:completion_once} assigns one designated completion slot, and
\eqref{eq:completion_consistency} prevents that slot from occurring
before all required DoFs have been processed. Because every $w_m>0$,
an optimal solution places $z_{m,t}=1$ at the earliest feasible
completion slot.

The decoding delay recovered from the completion indicator is
\begin{equation}
    \Delta_m
    =
    \sum_{t=a_m}^{H}
    (t-a_m+1)z_{m,t}.
    \label{eq:delay_from_completion}
\end{equation}
For a fixed packet trace, packet identities need not appear explicitly
in the mixed-integer program: under the rank-unit abstraction, the
cumulative-rank constraint \eqref{eq:arrival_causality} is sufficient
to prevent service from exceeding the innovative information
that has actually reached the destination.

This finite formulation is an exact benchmark for a selected stream
segment. It does not convert the operational model into a one-shot
batch system; MAIDS in Section~\ref{sec:general_scheduler} operates
directly on the dynamic active set $\mathcal M(t)$ without knowing a
terminal horizon.

\subsection{Computational Complexity under Exogenous Packet Availability}
\label{sec:complexity}

The main complexity result keeps the control boundary and the
time-varying packet availability of the general model, while removing
parallel processing and RLNC uncertainty. Consider the
following restricted destination-side availability pattern. For each
message $m$, choose an integer availability slot
$\tau_m^{\rm rel}$ and a packet set
$\mathcal P_m^{\rm rel}$ containing exactly $K_m$ innovative
coded packets. The value $\tau_m^{\rm rel}$ is determined by the
upstream packet-generation and delivery process; it is not a receiver
decision. Let
\begin{equation}
    \mathcal A_m(t)
    =
    \begin{cases}
        \mathcal P_m^{\rm rel},
        & t=\tau_m^{\rm rel},\\
        \emptyset,
        & \text{otherwise},
    \end{cases}
    \label{eq:batch_release_arrivals}
\end{equation}
set $r_m(a_m)=0$, and fix the receiver processing capacity to
\begin{equation}
    C=1.
    \label{eq:single_processing_unit}
\end{equation}
As a result, all innovative packets required by a message become available
together at the destination at an exogenous availability time, while
different messages may become available at different times. In this restricted construction, the message activation time is
$a_m=\tau_m^{\rm rel}$. The source generation time and the
destination-side availability time need not be the same; the latter is
the quantity relevant to scheduling. Once activated, a
message may be processed for several slots, interrupted while another
available message is processed, and resumed later.

This restricted problem is identical to the classical
preemptive single-machine problem
\begin{equation}
    1|r_j,\mathrm{pmtn}|\sum_j w_j C_j,
    \label{eq:preemptive_release_problem}
\end{equation}
which is strongly NP-hard for arbitrary job weights
\cite{schulz2002alphapoints}.

\begin{theorem}[Strong NP-hardness with one processing unit]
\label{thm:release_nphard}
The offline scheduling problem
\eqref{eq:offline_problem} is strongly NP-hard even when $C=1$, all arrivals are deterministic, the $K_m$ innovative packets required by
each message arrive together according to
\eqref{eq:batch_release_arrivals}, and processing is lossless and
deterministic.
\end{theorem}

\begin{proof}
Take an arbitrary instance of
$1|r_j,\mathrm{pmtn}|\sum_j w_j C_j$ with integer processing times
$p_j$, release dates $r_j$, and positive weights $w_j$. For every job
$j$, construct one message $m$ with
\begin{equation}
    K_m=p_j,
    \qquad
    \tau_m^{\rm rel}=r_j+1,
    \qquad
    w_m=w_j,
\end{equation}
and let $\mathcal P_m^{\rm rel}$ contain $K_m$ innovative
coded packets. Set $C=1$.

Beginning in slot $\tau_m^{\rm rel}$, each processed packet of message $m$
consumes exactly one unit of service. Because processing may
switch between messages at slot boundaries, the scheduler can preempt
and later resume any message exactly as the single-machine scheduler
can preempt and resume a job. The absolute decoding-completion slot
corresponds to the job completion time,
\begin{equation}
    T_m^{\rm dec}
    \longleftrightarrow
    C_j,
\end{equation}
while the decoding delay is
$\Delta_m=T_m^{\rm dec}-a_m+1=C_j-r_j$. Therefore
\begin{equation}
    \sum_m w_m\Delta_m
    =
    \sum_j w_j C_j
    -
    \sum_j w_j r_j.
\end{equation}
The second term is a constant fixed by the instance. Hence minimizing
weighted decoding delay is equivalent to minimizing the
classical weighted completion-time objective, and the two problems
have the same feasibility structure.
Since
$1|r_j,\mathrm{pmtn}|\sum_j w_j C_j$ is strongly NP-hard, the
restricted problem is strongly NP-hard. The general
problem contains this restriction and is therefore
strongly NP-hard as well.
\end{proof}

The theorem is deliberately stated with $C=1$, so the primary hardness
result does not rely on parallel processing. It also uses the simple
batch-release trace in which all packets of a message become available
together; arbitrary packet-by-packet exogenous arrivals only generalize
this availability constraint. Two classical boundaries are useful for
orientation. With a common availability epoch and $C=1$, the problem
reduces to Smith's ratio rule, while full initial availability with
input-dependent $C$ contains the fully parallel weighted-completion
problem of Zhang \emph{et al.} \cite{smith1956optimizers,zhang2013fullyparallel}.
The former correspondence is strengthened to progressive nonblocking
arrivals in Section~\ref{sec:maids_boundaries}; the latter is a secondary
parallel-capacity boundary rather than the main hardness mechanism.

\subsection{Exact Homogeneous Full-Buffer Benchmark}
\label{sec:homogeneous}

As an analytical reference, suppose all messages have the same
requirement $K_m=k$, all $k$ innovative packets for every message arrive at the beginning of slot~1, and processing is deterministic. Order the weights as
$w_{(1)}\geq\cdots\geq w_{(M)}$.

\begin{proposition}[Exact weighted delay under homogeneous full buffering]
\label{prop:homogeneous_exact}
Under the assumptions above,
\begin{equation}
    \mathrm{OPT}_{\mathrm{hom}}
    =
    \sum_{j=1}^{M}
    w_{(j)}
    \left\lceil
        \frac{jk}{C}
    \right\rceil.
    \label{eq:homogeneous_opt}
\end{equation}
\end{proposition}

\begin{proof}
Completing any $j$ messages requires processing at least $jk$ coded
packets, so the $j$th ordered completion time of every schedule is at
least $\lceil jk/C\rceil$. Pairing the largest weights with the
smallest attainable completion times gives
\eqref{eq:homogeneous_opt} as a lower bound. A work-conserving
schedule that processes messages in non-increasing weight order and
uses residual capacity in a completion slot immediately on the next
message attains every bound with equality.
\end{proof}

The proposition provides an exact reference for homogeneous full-buffer
instances. It is not a lower bound for arbitrary heterogeneous
requirements or delayed packet arrivals.

\section{Online Scheduling for the Heterogeneous Problem}
\label{sec:general_scheduler}

The offline problem requires future packet-availability information
and is strongly NP-hard even with a single processing unit under deterministic destination-side
packet availabilities. This
motivates a tractable online heuristic for the continuous message
stream. At each slot, the scheduler uses only information available
locally at the receiver: the current dynamic active-message set,
decoder rank states, and coded packets already stored in the receiver
buffers. MAIDS does not claim global optimality for arbitrary exogenous
arrival traces; instead, it uses a residual-work priority and solves the
resulting per-slot allocation problem exactly.

\subsection{Residual-Work Priority of MAIDS}

At the beginning of slot $t$, after adding the new arrivals, the
scheduler computes the buffered innovative availability $b_m(t)$ from
\eqref{eq:buffered_innovation}. Since $0\leq b_m(t)\leq d_m(t)$,
$b_m(t)$ is exactly the maximum useful service that message $m$ can
receive in the current slot. It is determined entirely by packets that
have already arrived, so MAIDS schedules realized innovative
information rather than predicted future packet contributions.

The weighted decoding-delay objective favors completing important
messages early, while heterogeneous messages may require different
amounts of remaining receiver work. Motivated by Smith's ratio rule
and its preemptive residual-work counterpart WSRPT, MAIDS assigns each
unfinished message the residual-work density
\begin{equation}
    \pi_m(t)
    =
    \frac{w_m}{d_m(t)},
    \qquad
    d_m(t)>0.
    \label{eq:message_priority}
\end{equation}
A positive priority is actionable only while message $m$ remains in
the dynamic active set and innovative information is currently buffered,
i.e., when $m\in\mathcal M(t)$ and $b_m(t)>0$. Therefore MAIDS
recomputes \eqref{eq:message_priority} every slot over the currently
serviceable messages; decoded messages leave the set and newly observed
messages enter it online.

The ratio in \eqref{eq:message_priority} should be interpreted as a
scheduling heuristic for the general exogenous-arrival problem, not
as an exact one-step minimizer of the original finite-horizon
objective. Its classical scheduling interpretation becomes exact only
in the special cases discussed in Section~\ref{sec:maids_relations}.

\subsection{Per-Slot MAIDS Allocation}

Given the current priorities and buffered innovative availability,
MAIDS solves the per-slot surrogate
\begin{equation}
\begin{aligned}
    \max_{\{y_m(t)\}}
    \quad&
    \sum_{m\in\mathcal M(t)}
    \pi_m(t)y_m(t)
    \\
    \mathrm{s.t.}\quad&
    \sum_{m\in\mathcal M(t)}y_m(t)\leq C,
    \\
    &
    0\leq y_m(t)\leq b_m(t),
    \qquad
    \forall m,
    \\
    &
    y_m(t)\in\mathbb Z_{\geq0}.
\end{aligned}
\label{eq:slot_processing}
\end{equation}
Each innovative DoF consumes one identical receiver processing unit,
and the objective in \eqref{eq:slot_processing} is linear. Hence the
per-slot problem is solved exactly by sorting the serviceable messages
in non-increasing order of $\pi_m(t)$ and greedily assigning processing
units up to $b_m(t)$. This exactness applies only to
\eqref{eq:slot_processing}, not to the original finite-horizon
weighted-delay problem. The allocation step requires
$O(|\mathcal M(t)|\log|\mathcal M(t)|)$ ordering time, apart from
linear-time buffer bookkeeping.

After the counts $\{y_m(t)\}$ are chosen, the receiver processes any
$y_m(t)$ packets from $\mathcal Q_m^+(t)$. Under the rank-unit
abstraction, each selected packet contributes one innovative DoF while
$d_m(t)>0$. Let $\mathcal P_m(t)\subseteq\mathcal Q_m^+(t)$ satisfy
\begin{equation}
    |\mathcal P_m(t)|=y_m(t).
    \label{eq:selected_packet_set}
\end{equation}
The next-slot receiver buffer is therefore
\begin{equation}
    \mathcal Q_m(t+1)
    =
    \mathcal Q_m^+(t)\setminus\mathcal P_m(t)
    \label{eq:buffer_update}
\end{equation}
for messages that remain unfinished; once $d_m(t+1)=0$, message $m$
and any remaining packets associated with it are removed from the
active receiver state.

\begin{algorithm}[!t]
\footnotesize
\DontPrintSemicolon
\SetCommentSty{footnotesize}

\textbf{Input:} receiver processing capacity $C$.\\
\textbf{State:} dynamic active-message set $\mathcal M(t)$, decoder
ranks $r_m(t)$, and receiver buffers $\mathcal Q_m(t)$.\\

Initialize $\mathcal M(1)\leftarrow\emptyset$.\;

\For{$t=1,2,\ldots$}{
    \tcp{Observe the exogenous message/packet stream}
    Observe newly received packet sets $\{\mathcal A_m(t)\}$ and
    newly observed message identities $\mathcal M^{\rm new}(t)$.\;

    \For{each $m\in\mathcal M^{\rm new}(t)$}{
        Initialize $r_m(t)\leftarrow0$ and
        $\mathcal Q_m(t)\leftarrow\emptyset$.\;
        Add $m$ to the active set $\mathcal M(t)$.\;
    }

    Set
    $\mathcal Q_m^+(t)\leftarrow
    \mathcal Q_m(t)\cup\mathcal A_m(t)$
    for every $m\in\mathcal M(t)$.\;

    \tcp{Compute current serviceable states}
    \For{each $m\in\mathcal M(t)$}{
        Compute $d_m(t)\leftarrow[K_m-r_m(t)]^+$.\;
        Compute $b_m(t)$ from \eqref{eq:buffered_innovation}.\;
        \eIf{$d_m(t)>0$ and $b_m(t)>0$}{
            $\pi_m(t)\leftarrow w_m/d_m(t)$\;
        }{
            $\pi_m(t)\leftarrow0$\;
        }
    }

    \tcp{Allocate receiver processing capacity}
    Sort serviceable messages by non-increasing $\pi_m(t)$.\;
    Greedily allocate at most $C$ processing units, with
    $y_m(t)\leq b_m(t)$ for each active message.\;

    \tcp{Select packets and update message states}
    \For{each $m\in\mathcal M(t)$}{
        Select any $\mathcal P_m(t)\subseteq\mathcal Q_m^+(t)$ with
        $|\mathcal P_m(t)|=y_m(t)$.\;
        Set $r_m(t+1)\leftarrow r_m(t)+y_m(t)$.\;
        Update $\mathcal Q_m(t+1)$ using
        \eqref{eq:buffer_update}.\;
        \If{$r_m(t+1)\geq K_m$}{
            Set $T_m^{\rm dec}\leftarrow t$ and
            $\Delta_m\leftarrow t-a_m+1$.\;
            Remove $m$ and its remaining buffered packets from the
            next-slot active state.\;
        }
    }

    Carry all remaining unfinished messages into $\mathcal M(t+1)$.\;
}
\caption{\small Streaming Message-Aware
Innovation-Deficit Scheduling (MAIDS).}
\label{alg:maids}
\end{algorithm}

\subsection{Relation to WSRPT and Smith's Rule}
\label{sec:maids_relations}

The single-processor batch-availability special case of
Theorem~\ref{thm:release_nphard} gives the clearest classical
interpretation of MAIDS. Set $C=1$, and suppose that all $K_m$
innovative packets required by each message become available together
at its destination-side availability time. Once message $m$ is
available, its remaining receiver work is exactly $d_m(t)$. MAIDS
therefore selects, among currently available unfinished messages, the
message with maximum
\begin{equation}
    \frac{w_m}{d_m(t)}.
\end{equation}
Under the mapping used in the NP-hardness proof,
$d_m(t)$ is the remaining processing time of the corresponding job.
Hence MAIDS reduces exactly to the Weighted Shortest Remaining
Processing Time (WSRPT) rule for
$1|r_j,\mathrm{pmtn}|\sum_j w_j C_j$. Batsyna \emph{et al.} study
this WSRPT rule as an efficient heuristic and establish a local
optimality property for its ordering
\cite{batsyna2014online}. This equivalence does not imply global
optimality when availability times differ.

With common full availability and $C=1$, the same residual-ratio
rule starts from $w_m/K_m$ and keeps the selected message at increasing
priority as its deficit decreases; it therefore recovers Smith's
non-increasing $w_m/K_m$ order \cite{smith1956optimizers}. The formal
result in Section~\ref{sec:maids_boundaries} extends this intuition to
progressive nonblocking arrivals. Under full availability and $C>1$,
the rule is also consistent with the fully parallel residual-ratio
packing interpretation \cite{zhang2013fullyparallel}; when all
$K_m=k$, it attains the homogeneous benchmark in
Proposition~\ref{prop:homogeneous_exact}.

Under general packet-by-packet exogenous arrivals, these classical
relations no longer characterize the global optimum. MAIDS instead
recomputes the residual-work ratio only over messages with currently
buffered innovative information, so a high-priority message cannot
consume processing service before its coded information has actually
arrived.

\subsection{Performance Boundaries of MAIDS}
\label{sec:maids_boundaries}

The preceding classical correspondences identify two sources of
online difficulty in the general problem: heterogeneous
message priorities and progressive packet availability. The following
three boundary results separate cases in which MAIDS is exactly optimal
from the unrestricted weighted online regime, where no constant
competitive guarantee is possible for any deterministic scheduler.

The first condition identifies progressive arrival traces that do not
starve a single-unit receiver.

\begin{definition}[Nonblocking arrival trace]
\label{def:receiver_nonblocking}
For $C=1$, a finite trace is \emph{nonblocking} if, for every
message $m$ and every slot $t\geq a_m$,
\begin{equation}
    R_m^{\rm arr}(t)
    \geq
    \min\{K_m,\,t-a_m+1\}.
    \label{eq:receiver_nonblocking}
\end{equation}
Thus, from activation onward, innovative information for message $m$
arrives at least as fast as a unit-rate processor could continuously
process that message. This condition allows progressive arrivals and
is strictly weaker than requiring all $K_m$ packets to be available at
activation.
\end{definition}

\begin{theorem}[Equal-weight optimality]
\label{thm:maids_equal_weight}
Consider a finite trace-conditioned instance with $C=1$, a
nonblocking arrival trace, and equal positive message weights
$w_m=w$ for all $m$. Then MAIDS minimizes the total decoding delay
$\sum_m \Delta_m$.
\end{theorem}

\begin{proof}
At the beginning of slot $t$, any unfinished message $m$ has received
at most one unit of service in each earlier slot since its
activation, and hence
$r_m(t)\leq t-a_m$. If $t-a_m+1\leq K_m$, condition
\eqref{eq:receiver_nonblocking} gives
$R_m^{\rm arr}(t)\geq t-a_m+1>r_m(t)$; if
$t-a_m+1>K_m$, it gives $R_m^{\rm arr}(t)=K_m>r_m(t)$ for every
unfinished $m$. Therefore every active unfinished message has at least
one useful packet available whenever it could be selected, so packet
availability never restricts the single-unit scheduling decision.

The resulting problem is equivalent to preemptive
single-machine scheduling with release epoch $r_j=a_m-1$, processing
requirement $p_j=K_m$, and equal weights. Since
$\pi_m(t)=w/d_m(t)$, MAIDS selects the message with the smallest
remaining processing requirement and therefore coincides with the
Shortest Remaining Processing Time (SRPT) discipline. SRPT minimizes
total flow/completion time with release dates under preemption
\cite{schrage1968srpt}. Identifying the classical completion epoch
with $C_j=T_m^{\rm dec}$ gives
$C_j-r_j=T_m^{\rm dec}-a_m+1=\Delta_m$, so the same schedule
minimizes $\sum_m\Delta_m$.
\end{proof}

\begin{corollary}[Common-activation weighted optimality]
\label{cor:maids_common_activation}
Consider $C=1$ and a nonblocking arrival trace. If all
messages have a common activation time $a_m=a$ but may have arbitrary
positive weights, then MAIDS is optimal and completes messages in
non-increasing $w_m/K_m$ order.
\end{corollary}

\begin{proof}
By Definition~\ref{def:receiver_nonblocking}, packet availability never
blocks an unfinished message. At the common activation time, MAIDS
orders messages by $w_m/K_m$. Once a message $m$ is selected, each unit
of service decreases $d_m(t)$ while leaving every unserved message's
deficit unchanged. Hence $w_m/d_m(t)$ can only increase relative to the
ratios of the other waiting messages, so MAIDS continues serving $m$
until completion. The resulting nonpreemptive completion order is
therefore non-increasing $w_m/K_m$, which is Smith's optimal ratio rule
\cite{smith1956optimizers}. Full initial availability is a special case of the
nonblocking condition, so this corollary extends the
full-availability Smith-rule interpretation in
Section~\ref{sec:maids_relations} to progressive arrivals.
\end{proof}

The preceding results are exact positive boundaries. The next result
shows why a universal constant competitive ratio should not be
expected once arbitrary weights and arbitrary online activations are
allowed.

\begin{proposition}[No constant deterministic competitive ratio]
\label{prop:no_constant_online}
For the unrestricted weighted online problem with $C=1$, no
deterministic scheduling policy has an $O(1)$
competitive ratio for total weighted decoding delay, even in the
special case where all $K_m$ required packets of each message are
available at its activation time.
\end{proposition}

\begin{proof}
Consider the classical online preemptive single-machine weighted-flow-
time problem. For each job $j$ with release time $r_j$, processing
requirement $p_j$, and weight $w_j$, construct a message with
$a_m=r_j+1$, $K_m=p_j$, and $w_m=w_j$, and make all $K_m$ packets
available at the beginning of slot $a_m$. One unit of service
then corresponds exactly to one unit of preemptive machine processing.
If the message completes in slot $T_m^{\rm dec}$, identify the
classical completion epoch with $C_j=T_m^{\rm dec}$. Its decoding delay
is
\begin{equation}
    \Delta_m
    =T_m^{\rm dec}-a_m+1
    =C_j-r_j,
\end{equation}
which is exactly the classical flow time. Thus any deterministic
$O(1)$-competitive scheduler for this special case would
imply a deterministic $O(1)$-competitive algorithm for online
preemptive single-machine weighted flow time. Bansal and Chan proved
an $\omega(1)$ lower bound on the competitive ratio of every
deterministic online algorithm for that problem
\cite{bansal2009weightedflow}, yielding the claim.
\end{proof}

Proposition~\ref{prop:no_constant_online} is a boundary for the entire
online problem rather than a limitation specific to MAIDS.
Together with Theorem~\ref{thm:maids_equal_weight} and
Corollary~\ref{cor:maids_common_activation}, it identifies a clear
frontier: MAIDS is exactly optimal when either equal weights reduce its
priority to SRPT under nonblocking arrivals or common activation
reduces it to Smith's ratio rule, whereas the fully heterogeneous
online weighted-delay problem does not admit a universal constant
competitive guarantee for any deterministic policy.
\section{Simulation Results}
\label{sec:simulation}

The evaluation consists of three complementary parts. First, the general
streaming experiment tests the online scheduler under exogenous
packet-by-packet arrivals against two online baselines. Second, the finite
batch-release experiment instantiates the NP-hard special case of
Section~\ref{sec:complexity} and compares MAIDS with WSRPT and the exact
offline solution of \eqref{eq:offline_problem}. Third, the
performance-boundary validation checks two exact regimes from
Section~\ref{sec:maids_boundaries}: an
equal-weight nonblocking progressive-arrival regime, where
MAIDS should coincide with SRPT, and a common-full-availability regime,
which is the full-buffer endpoint of the common-activation weighted
boundary and should coincide with Smith/WSPT. The impossibility boundary
of Proposition~\ref{prop:no_constant_online} is a worst-case analytical
statement and is therefore not treated as a Monte Carlo validation.

\subsection{Simulation Setup}

All simulations are slot-based and preserve the control
boundary of Section~\ref{sec:system_model}. The upstream process first
generates an exogenous coded-packet arrival trace; the same realized
trace is then replayed under all policies being compared.
Hence a scheduling decision can change only the order in which
already-arrived packets are processed and cannot change future packet
arrivals. The rank-unit abstraction is retained, so every successfully
received coded packet contributes one innovative DoF until its message
reaches rank $K_m$. This isolates scheduling from finite-field
dependence and makes the simulated service unit identical to that in
\eqref{eq:rx_processing_capacity}.

For a completed message cohort $\mathcal C$, the reported performance
metric is the weighted mean decoding delay
\begin{equation}
    \overline{\Delta}_w(\mathcal C)
    =
    \frac{
        \sum_{m\in\mathcal C} w_m\Delta_m
    }{
        \sum_{m\in\mathcal C} w_m
    }.
    \label{eq:sim_weighted_mean_delay}
\end{equation}
This is the normalized form of the objective in
\eqref{eq:weighted_delay_objective}; it keeps results comparable across
random traces containing different numbers and mixtures of messages.

Table~\ref{tab:simulation_parameters} summarizes the main parameters.
For streaming, $\rho/C$ is the long-run innovative-packet workload
offered to the receiver, normalized by processing capacity, and is
varied from $0.40$ to $0.95$ in steps of $0.05$; hence all steady-state
experiments satisfy $\rho/C<1$. For the finite batch and common-full
benchmarks, the same numerical grid is denoted by
$x\equiv\rho/C$ only as a normalized workload-control parameter, not as
a steady-state offered load.
New-message requests follow a Poisson process with
\begin{equation}
    \lambda_{\rm msg}
    =
    \frac{\rho}{\mathbb E[K]},
    \qquad
    \mathbb E[K]=42,
    \label{eq:sim_message_rate}
\end{equation}
so increasing $\rho$ increases the number of overlapping messages
without increasing the packet-generation speed of an already active
message. Each source assigned to a message independently emits its next
coded packet with probability $p_s=0.5$ in each slot. A source remains
locked to its current message until its assigned packet quota is
exhausted, after which it may support a later message. Thus packets of
successive messages do not interleave at an individual source, while a
single message may simultaneously receive packets from multiple
sources.

Across all streaming settings, the mean realized $\rho/C$ differs from
the target value by at most $0.0132$. The selected streaming parameters
also keep the nominal packet workload below the aggregate source
transmission opportunity $Np_s$: the largest target workload is
$\rho=4\times0.95=3.8$ packets/slot, whereas $Np_s=12\times0.5=6$ is the aggregate expected transmission
rate when all sources are active.

\begin{table}[!t]
\centering
\caption{Simulation Parameters}
\label{tab:simulation_parameters}
\scriptsize
\begin{tabular}{@{}ll@{}}
\toprule
\textbf{Parameter} & \textbf{Value} \\
\midrule
\multicolumn{2}{c}{\emph{General streaming}} \\
Receiver capacity $C$ & $\{1,2,4\}$ packets/slot \\
Normalized load $\rho/C$ & $0.40:0.05:0.95$ \\
Message weights $w_m$ & $\{1,2,4,8,16\}$, uniform \\
Upstream sources $N$ & $12$ \\
Generation size $K_m$ & $U\{4,\ldots,80\}$ \\
Support size $G_m$ & $\{1,2,3\}$ \\
$\Pr(G_m=1,2,3)$ & $(0.50,0.35,0.15)$ \\
Per-source send probability $p_s$ & $0.5$ \\
Warm-up / measurement / guard & $5000/15000/15000$ slots \\
Monte Carlo trials & $96$ per $(C,\rho/C)$ \\
\midrule
\multicolumn{2}{c}{\emph{NP-hard batch release}} \\
Receiver capacity $C$ & $\{1,2,4\}$ packets/slot \\
Workload parameter $x\equiv\rho/C$ & $0.40:0.05:0.95$ \\
Message weights $w_m$ & $\{1,2,4,8,16\}$, uniform \\
Messages per instance $M$ & $8$ \\
Generation size $K_m$ & $U\{2,\ldots,10\}$ \\
Release-window fraction $\beta$ & $0.35$ \\
Monte Carlo trials & $16$ per $(C,x)$ \\
Exact MILP time limit & $30$ s per instance \\
\midrule
\multicolumn{2}{c}{\emph{Equal-weight nonblocking boundary}} \\
Receiver capacity / messages & $C=1$, $M=10$ \\
Message weights & $w_m=1$ for all $m$ \\
Generation size $K_m$ & $U\{2,\ldots,20\}$ \\
Activation time $a_m$ & $U\{1,\ldots,12\}$ \\
Packet availability & one packet/slot from $a_m$ \\
Random instances & $500$ \\
\midrule
\multicolumn{2}{c}{\emph{Common full availability}} \\
Receiver capacity / messages & $C=1$, $M=5$ \\
Workload parameter $x\equiv\rho/C$ & $0.40:0.05:0.95$ \\
Message weights $w_m$ & $\{1,2,4,8,16\}$, uniform \\
Reference horizon & $24$ slots \\
Monte Carlo trials & $40$ per $x$ value \\
\bottomrule
\end{tabular}
\end{table}

For streaming, a trial contains a $5000$-slot warm-up, a $15000$-slot
measurement interval, and a $15000$-slot guard interval. The metric in
\eqref{eq:sim_weighted_mean_delay} is computed for messages activated
during the measurement interval; the guard interval allows those
messages to finish without truncating their delays. Within a fixed
trial and load, MAIDS and all baselines use exactly the same exogenous
arrival trace.

The positive performance boundaries use finite instances. In the
equal-weight case, the parameters in Table~\ref{tab:simulation_parameters}
are combined with one progressive arrival per slot from activation until
completion of the message's $K_m$ arrivals:
\begin{equation}
    |\mathcal A_m(t)|
    =
    \begin{cases}
        1, & a_m\leq t\leq a_m+K_m-1,\\
        0, & \text{otherwise}.
    \end{cases}
    \label{eq:sim_equal_weight_arrival}
\end{equation}
Thus $R_m^{\rm arr}(t)=\min\{K_m,\max(0,t-a_m+1)\}$ and the
nonblocking condition holds with equality. The common-full validation
instead makes all required packets available at slot~1 and compares
MAIDS with Smith/WSPT.

\subsection{General Streaming Performance}

MAIDS is compared with two online baselines. Static-WSPT
uses the fixed index $w_m/K_m$ throughout the lifetime of message $m$,
whereas Weight-only uses $w_m$ and ignores message size and remaining
work. MAIDS instead recomputes $w_m/d_m(t)$ every slot and therefore
adapts its priority as a message approaches decoding.

\begin{figure}[!t]
\centering
\subfigure[$C=1$]{
    \includegraphics[width=0.96\columnwidth]{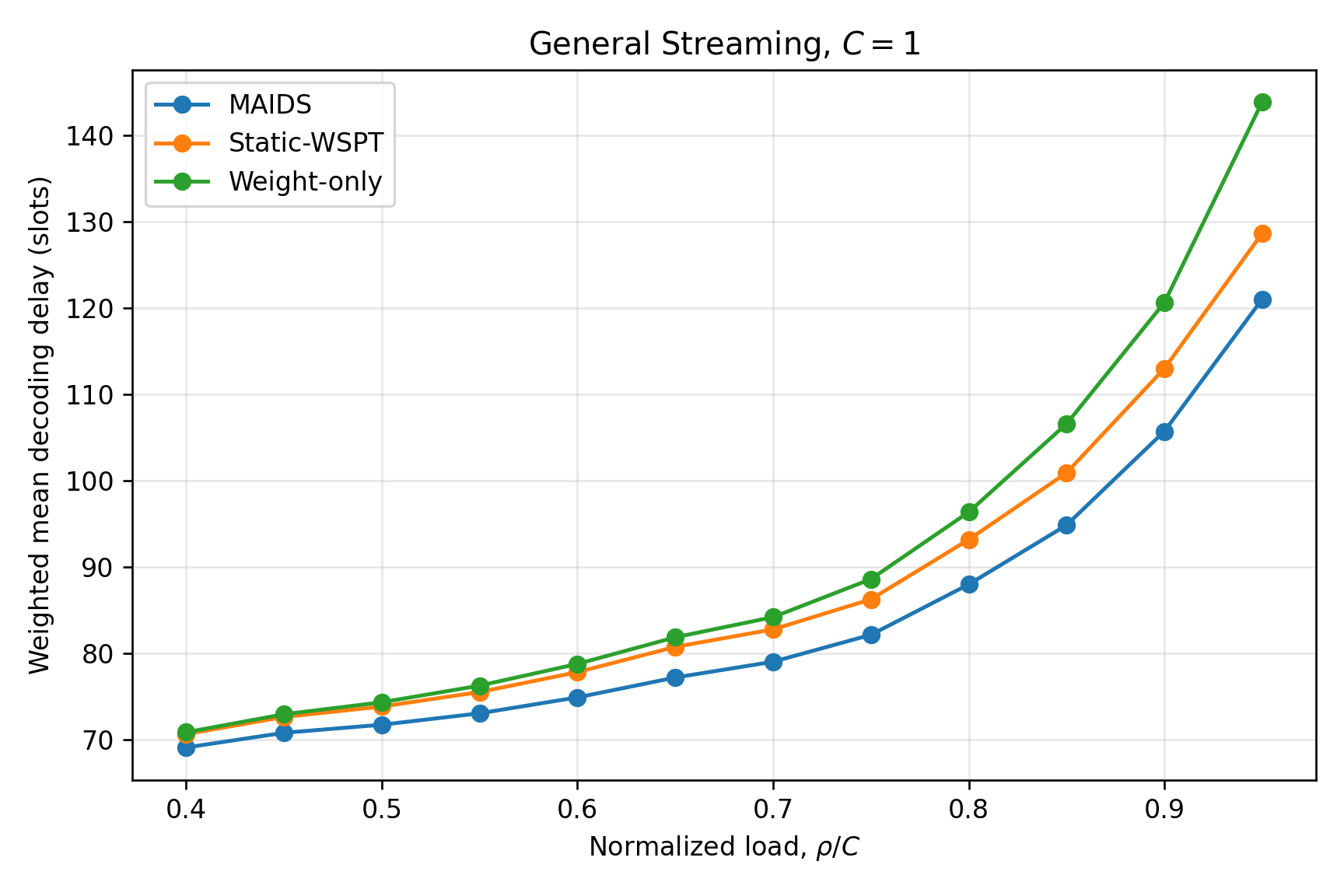}
    \label{fig:streaming_c1}
}\\[-0.5ex]
\subfigure[$C=2$]{
    \includegraphics[width=0.96\columnwidth]{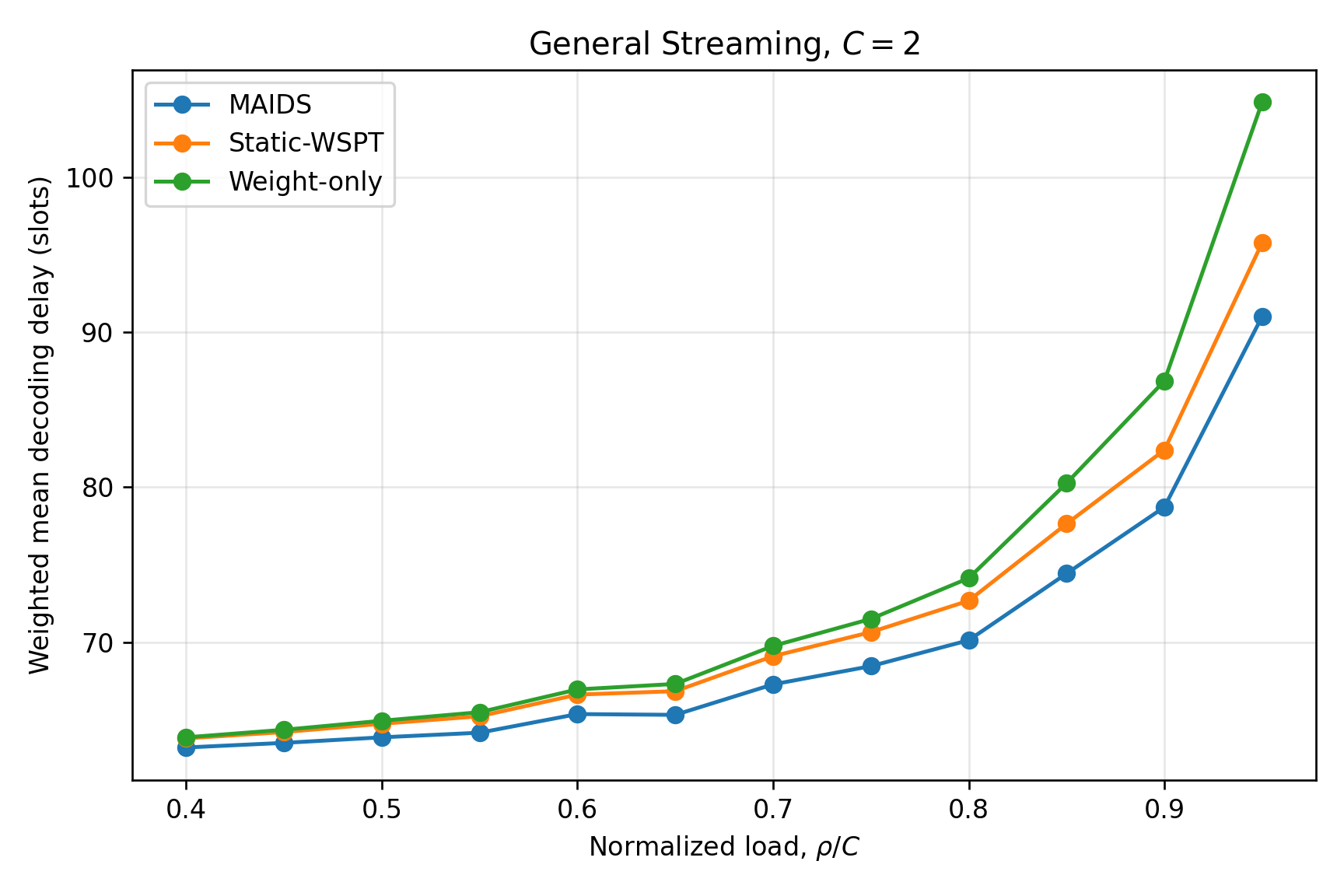}
    \label{fig:streaming_c2}
}\\[-0.5ex]
\subfigure[$C=4$]{
    \includegraphics[width=0.96\columnwidth]{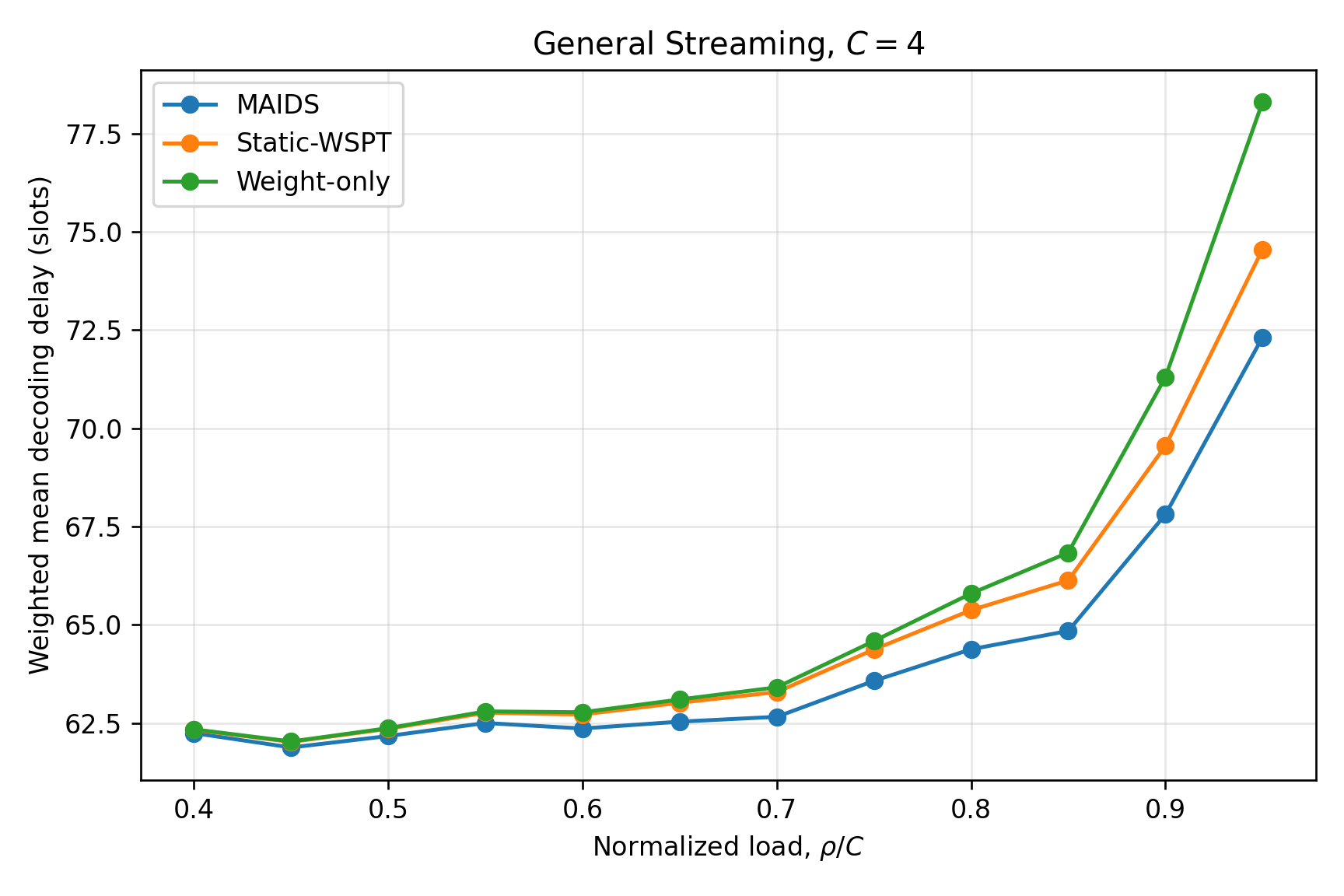}
    \label{fig:streaming_c4}
}
\caption{General streaming performance under exogenous coded-packet
arrivals.}
\label{fig:streaming_results}
\end{figure}

Fig.~\ref{fig:streaming_results} shows the weighted mean decoding delay
for $C=1,2,$ and $4$. In all three cases, delay increases with
$\rho/C$ as more messages overlap and compete for the finite
processing budget. MAIDS consistently yields the smallest delay among
the three online policies. Because all policies replay the same trace
within each trial, paired 95\% confidence intervals are computed for the
baseline-minus-MAIDS delay
differences over the 96 trials at each setting. The intervals are
strictly positive at all 36 tested $(C,\rho/C)$ settings for both
baselines. The smallest separation occurs at $C=4$ and $\rho/C=0.40$:
the mean Static-WSPT--MAIDS difference is $0.0853$ slots with a 95\%
confidence interval of $[0.0756,0.0951]$, while the corresponding
Weight-only--MAIDS difference is $0.0961$ slots with a 95\% confidence
interval of $[0.0853,0.1069]$. At light load, the
policies are relatively close because receiver contention is limited.
The separation becomes more pronounced as $\rho/C$ approaches one,
where completion order has a larger effect on queueing and decoding
delay.

The behavior across $C$ also illustrates the role of receiver
capacity. For $C=4$, the three policies remain close over a wider
low-load region because multiple buffered packets can be processed in
the same slot. Once the load becomes high, however, the curves separate
again and MAIDS retains the lowest delay. These results indicate that
the dynamic residual-work term in $w_m/d_m(t)$ is most useful when
several heterogeneous messages are simultaneously serviceable and the
receiver must decide which nearly completed messages should receive
scarce processing capacity first.

\subsection{NP-Hard Batch-Release Benchmark}

The NP-hard benchmark evaluates MAIDS on finite batch-release instances
derived from the construction in Section~\ref{sec:complexity}. Each
instance contains $M=8$ messages with
$K_m\sim U\{2,\ldots,10\}$ and weights drawn from
$\{1,2,4,8,16\}$. For a realized total workload
\begin{equation}
    W=\sum_m K_m
\end{equation}
and finite workload-control parameter $x\equiv\rho/C$, the reference
horizon is defined as
\begin{equation}
    H_x
    =
    \left\lceil
        \frac{W}{Cx}
    \right\rceil.
    \label{eq:sim_batch_horizon}
\end{equation}
Here $x\in\{0.40,0.45,\ldots,0.95\}$ controls release-window
compression for this finite instance; unlike streaming $\rho/C$, it is
not interpreted as a steady-state offered load. All $K_m$ packets of
message $m$ become available together at a random
destination-side release time
\begin{equation}
    \tau_m^{\rm rel}
    \sim
    U\left\{
        1,\ldots,
        \left\lceil\beta H_x\right\rceil
    \right\},
    \qquad
    \beta=0.35.
    \label{eq:sim_batch_release}
\end{equation}
The same $(K_m,w_m)$ realization is reused across $x$ values within a
Monte Carlo trial, while changing $x$ changes the release window and
therefore the amount of scheduling overlap.

For every instance, MAIDS, WSRPT, and the exact offline optimum
obtained from the mixed-integer formulation are evaluated. For $C=1$, the
batch-availability model makes $d_m(t)$ exactly equal to the remaining
processing time, so MAIDS is classical WSRPT as established in
Section~\ref{sec:maids_relations}. For $C>1$, the plotted WSRPT curve
denotes the natural capacity-$C$ extension that applies the same
weighted-shortest-remaining-work ratio and allocates up to $C$ units
per slot. Consequently, in the batch experiments it uses exactly the
same per-slot ordering as MAIDS. Numerically, the maximum absolute
difference between the MAIDS and WSRPT objectives over all evaluated
batch instances is zero. Accordingly, the WSRPT curve is completely
superimposed on the MAIDS curve in all three panels of
Fig.~\ref{fig:nphard_results}.

Fig.~\ref{fig:nphard_results} compares the resulting delays with Exact
OPT. For $C=1$, the three curves are nearly indistinguishable, although
WSRPT/MAIDS is not globally optimal when release times differ. The
difference remains small on average for $C=2$, while a more visible
separation from the offline optimum appears for $C=4$. The exact
benchmark exploits complete future knowledge of all message release
times, whereas MAIDS remains an online local policy.

\begin{figure}[!t]
\centering
\subfigure[$C=1$]{
    \includegraphics[width=0.96\columnwidth]{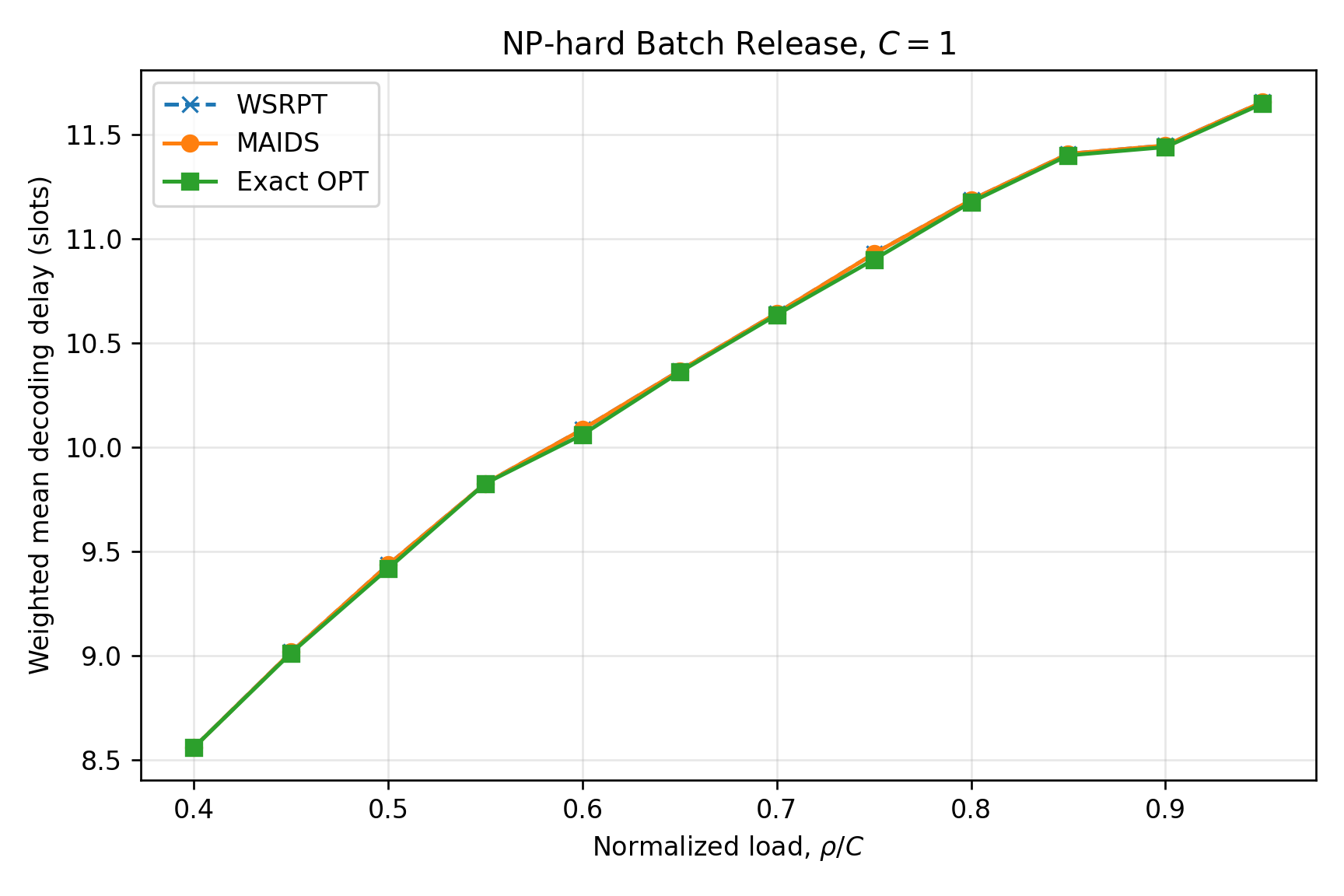}
    \label{fig:nphard_c1}
}\\[-0.5ex]
\subfigure[$C=2$]{
    \includegraphics[width=0.96\columnwidth]{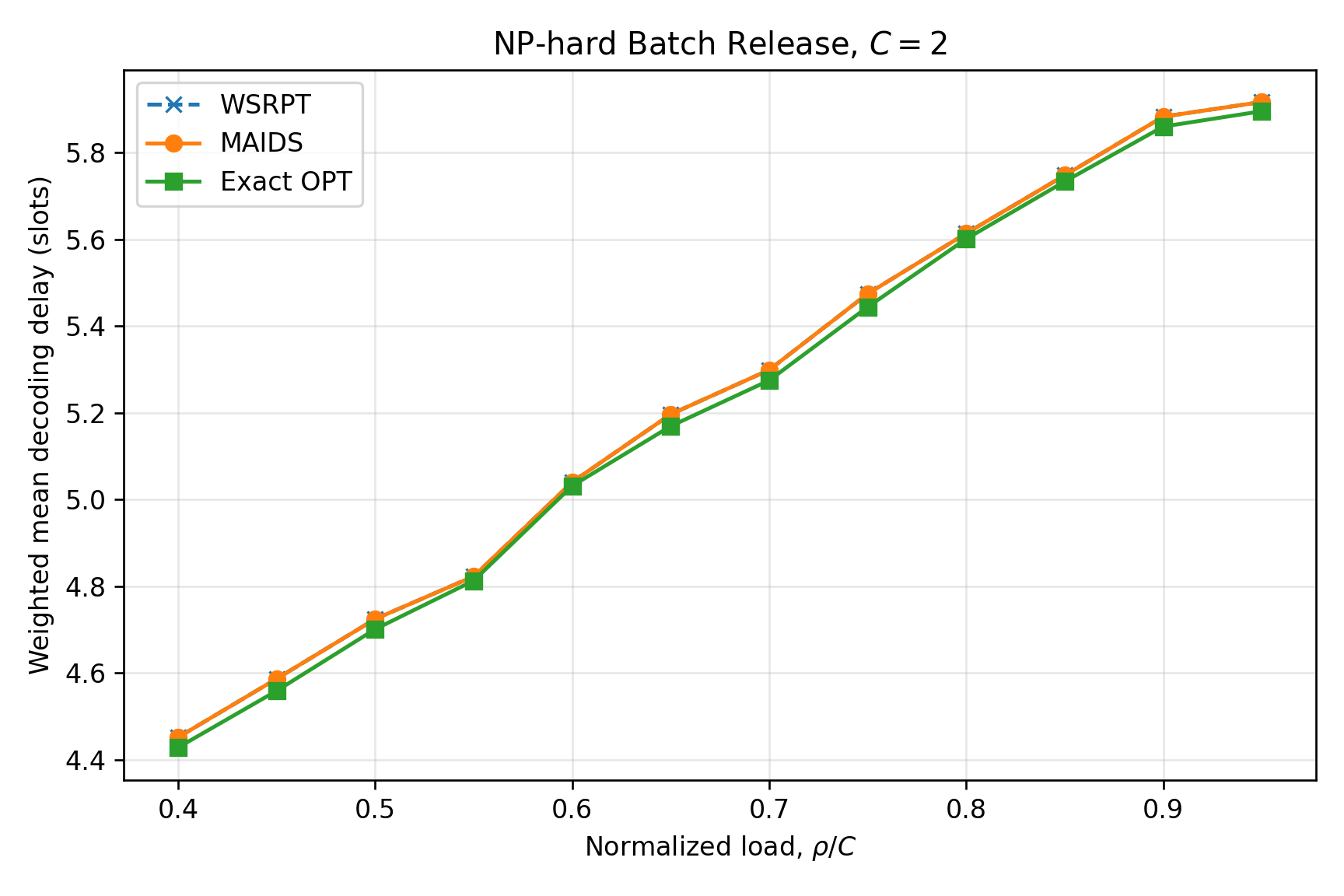}
    \label{fig:nphard_c2}
}\\[-0.5ex]
\subfigure[$C=4$]{
    \includegraphics[width=0.96\columnwidth]{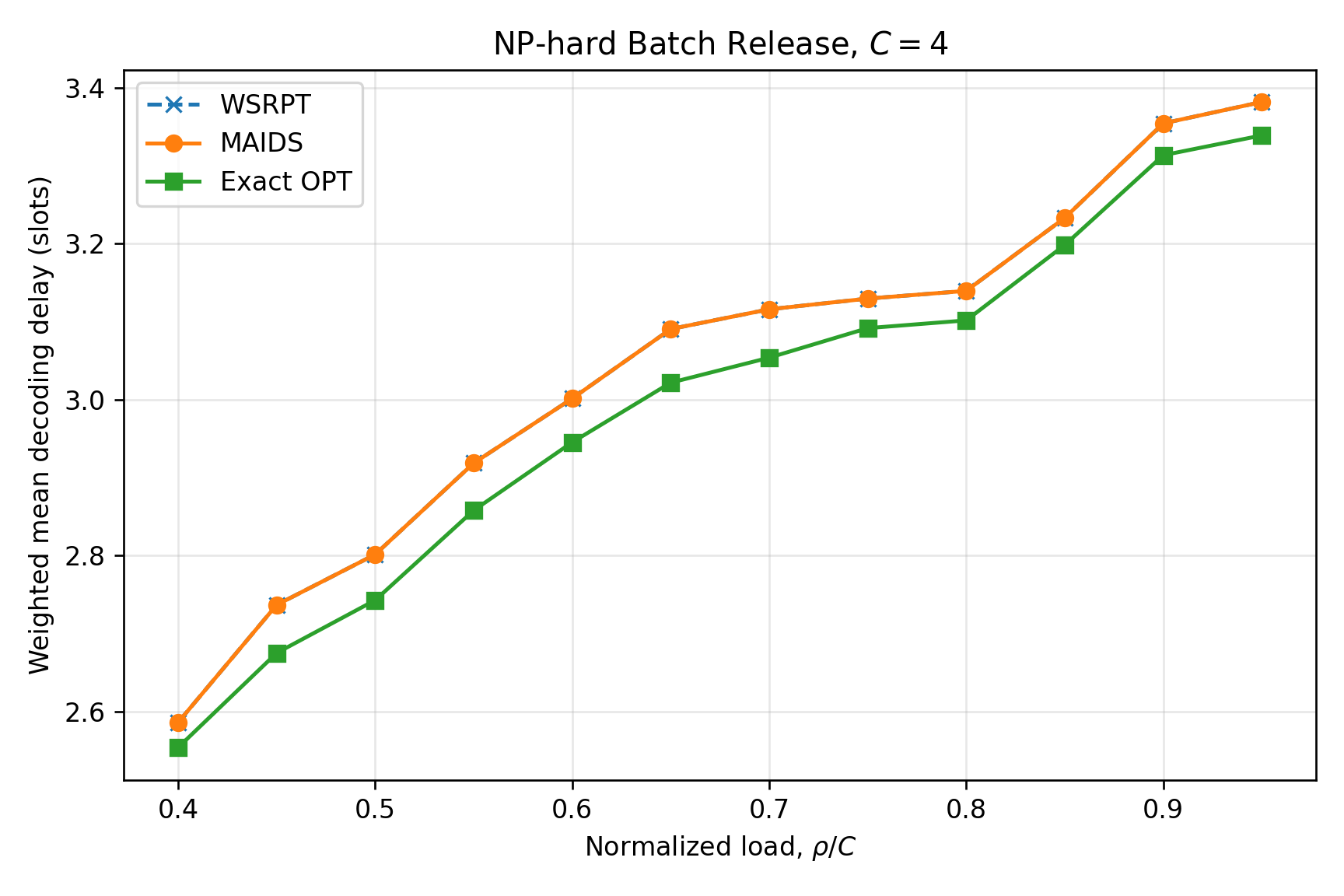}
    \label{fig:nphard_c4}
}
\caption{NP-hard batch-release benchmark.}
\label{fig:nphard_results}
\end{figure}

To quantify differences that are difficult to see in the absolute-delay
plots, define the paired-instance relative optimality gap
\begin{equation}
    g
    =
    \frac{
        J_{\rm MAIDS}-J_{\rm OPT}
    }{
        J_{\rm OPT}
    }
    \times 100\%.
    \label{eq:sim_relative_gap}
\end{equation}
Because smaller delay is better, $g=0$ indicates an optimal MAIDS
schedule and positive $g$ is the percentage by which its weighted delay
exceeds the offline optimum. Table~\ref{tab:nphard_gap} aggregates the
gap over all 12 $x$ values and 16 random instances per value. The mean
gap is only $0.10\%$ for $C=1$ and $0.41\%$ for $C=2$. For $C=4$ it
increases to $1.86\%$, showing that MAIDS remains close to the exact
benchmark on average but should not be interpreted as globally
optimal. The maximum observed gaps also increase with $C$, reaching
$16.36\%$ for the most difficult evaluated $C=4$ instance.
All 576 MILP instances are solved successfully within the 30-s time
limit, so every value labeled Exact OPT in this benchmark corresponds
to a solved offline optimum. The last column of
Table~\ref{tab:nphard_gap} reports the fraction of instances on which
MAIDS itself attains that optimum.

\begin{table}[!t]
\centering
\caption{Relative Optimality Gap of MAIDS on NP-Hard Batch Instances}
\label{tab:nphard_gap}
\scriptsize
\begin{tabular}{@{}cccc@{}}
\toprule
$C$ &
Mean gap (\%) &
Maximum gap (\%) &
MAIDS-optimal instances (\%) \\
\midrule
1 & 0.098 & 3.723 & 91.67 \\
2 & 0.409 & 5.861 & 77.60 \\
4 & 1.862 & 16.364 & 50.52 \\
\bottomrule
\end{tabular}
\end{table}

\subsection{Performance Boundary Validation}

The two positive exact boundaries in
Section~\ref{sec:maids_boundaries} are validated as implementation-level
consistency checks rather than new heuristic comparisons.

\subsubsection{Equal-weight nonblocking arrivals}
Under \eqref{eq:sim_equal_weight_arrival}, arrivals are progressive but
nonblocking. Across 500 independent instances, MAIDS and SRPT have zero
objective difference in every instance, and their complete message
completion-time vectors are identical. This matches
Theorem~\ref{thm:maids_equal_weight}.

\subsubsection{Common full availability}
For the common-full endpoint of
Corollary~\ref{cor:maids_common_activation}, 480 heterogeneous instances
are evaluated with all required packets available at slot~1. The
maximum absolute difference between MAIDS and the Smith/WSPT optimum is
$0$, recovering the predicted full-availability boundary.

Thus the simulations validate the equal-weight progressive boundary and
the full-buffer endpoint of the common-activation boundary. The
negative result in Proposition~\ref{prop:no_constant_online} is a
worst-case impossibility theorem and does not require Monte Carlo
validation.

\section{Conclusion}

This paper studies weighted decoding-delay scheduling for continuous
RLNC-coded multi-source traffic with exogenous packet arrivals and finite
processing capacity. The trace-conditioned offline problem is strongly
NP-hard in a batch-release subclass, while MAIDS provides a tractable
online rule based on message weight and remaining decoding deficit. Exact
optimality holds in two structured single-unit regimes, whereas the
unrestricted weighted online setting admits no universal deterministic
$O(1)$ competitive guarantee.

Simulation results confirm the predicted exact performance boundaries and
show that MAIDS performs favorably against the tested online baselines.
Across the finite NP-hard benchmark, MAIDS also remains close to the exact
offline optimum on average, supporting its use as a practical scheduler
for heterogeneous coded-message streams.

\bibliographystyle{IEEEtran}
\bibliography{Ref}

\end{document}